\documentclass[11pt]{article}
\usepackage[margin=1in]{geometry}
\usepackage{amsmath,amssymb,amsthm}
\usepackage{graphicx}
\usepackage{booktabs}
\usepackage{tabularx}
\usepackage{natbib}
\usepackage{hyperref}
\usepackage{authblk}
\usepackage[activate={true,nocompatibility},final,tracking=true,factor=1100,stretch=10,shrink=10]{microtype}
\newtheorem{proposition}{Proposition}
\newtheorem{assumption}{Assumption}

\title{Stable Transport Meta-Analysis for Heterogeneous Cardiovascular Trials:\\
A Nuisance-Anchor Framework with a Sign-Stability Diagnostic}

\author[1]{Ibrahim Halil Tanboga, MD, PhD\thanks{Corresponding author: \texttt{haliltanboga@gmail.com}}}
\affil[1]{Department of Biostatistics, Nisantasi University Medical School, Istanbul, Turkey}

\date{\today}

\begin{document}
\maketitle

\begin{abstract}
Random-effects meta-analysis summarises heterogeneous trials by estimating
an average effect over the observed evidence base. In cardiovascular
medicine, trial effects vary systematically across treatment era, endpoint
definition, background therapy and baseline case-mix, and the clinically
relevant target is rarely the historical average. We propose \emph{stable
transport meta-analysis} (AMT-MA), a nuisance-anchor estimator that models
anchor-aligned variation as a regression coefficient but does not transport
it to the target population. The estimator blends a weighted-average loss
with a scale-normalised softmax regime loss, and is paired with a
precision-weighted trial-level sign-stability diagnostic and a two-condition
abstention rule. AMT-MA is not designed to dominate random-effects pooling
in RMSE; it is designed to change the estimand from a historical average
effect to a stable target-population effect, and to refuse a single pooled
answer when such a stable effect is not identified. In a pre-specified
ADEMP simulation across six scenarios, AMT-MA at the recommended blend weight
$\rho = 0.2$ was less efficient than fixed-effect pooling under true
homogeneity but comparable to WLS meta-regression, reduced target-effect
bias relative to unadjusted pooling, and improved target-effect coverage
in adversarial regimes where classical Wald intervals break down
(dominant historical mega-trial: AMT-MA perturbation coverage $0.85$
versus fixed-effect Wald coverage $0.01$; confounded anchor: $0.86$ vs.\
$0.34$; anchor shift: $0.91$ vs.\ $0.60$); WLS meta-regression remained
competitive in coverage when its moderator--anchor specification was
correctly specified. Under sign-flip heterogeneity the two-condition
abstention rule triggered in approximately $84\%$ of replications,
compared with approximately $28$--$30\%$ in stable and anchor-shift
regimes — substantial but imperfect separation between sign-conflict and
stable-effect settings. We illustrate the method on seventy
post-myocardial-infarction streptokinase trials (Olkin 1995) and, as an
illustrative stress test, on five primary-prevention aspirin trials from
a contemporary-versus-older dataset. In each case AMT-MA quantifies the
uncertainty of transport to a modern target population, providing a
clinically interpretable alternative to averaging heterogeneous trial
effects.
\end{abstract}

\section{Introduction}
\label{sec:intro}

Meta-analysis is the standard framework for synthesizing evidence from
multiple randomized trials. In cardiovascular medicine, however, trial effects
vary systematically across treatment era, endpoint definition, geography,
adjudication procedure, and patient case-mix
\citep{higgins2009predictive}. Classical fixed- and random-effects pooling
\citep{dersimonian1986meta,viechtbauer2010metafor} answer the question: \emph{what
was the average treatment effect across historical trials?} Modern care decisions
instead ask: \emph{what component of the treatment effect is stable enough to
be transported to the target population in which the decision will be made?}

Recent work on causally interpretable meta-analysis \citep{dahabreh2020transport}
has redefined the estimand from ``average observed trial effect'' toward a
target-population average effect, but the summarization step typically remains
within the random-effects paradigm. Meanwhile, the multi-source learning
literature offers three complementary ideas that have not been combined into
a meta-analytic method: anchor regression \citep{rothenhausler2021anchor},
maximin effects \citep{meinshausen2015maximin,guo2024statistical}, and
group-distributionally robust optimization \citep{wang2023drmultisource}.

We propose \textbf{AMT-MA: Anchor--Maximin--Transport Meta-Analysis}. AMT-MA
is an estimand shift first and a new algorithm second. Given study-level
effects $\{y_i,v_i\}$, transportable moderators $Z_i$, anchor variables
$A_i$, and regime labels $g_i$, we estimate the stable target effect
$\theta_{\mathrm{stable}}(Q) = \bar Z_Q^{\top}\beta^{\star}$, where
$(\beta^{\star}, \gamma^{\star})$ is obtained from a nuisance-anchor
objective that models anchor-aligned variation through $A^{\top}\gamma$,
uses a softmax regime loss for robustness, and transports only the
$Z^{\top}\beta^{\star}$ component to the target population. This paper formalizes the estimand (Section~\ref{sec:estimand}),
specifies the estimator (Section~\ref{sec:method}), evaluates it via an ADEMP
simulation study \citep{morris2019ademp} with six scenarios
(Section~\ref{sec:sim}), and applies it to the Olkin 1995 streptokinase post-MI
meta-analysis \citep{olkin1995meta} (Section~\ref{sec:application}). The
novelty is not a new $\tau^{2}$ estimator; it is a clinically interpretable
target-population estimand that is robust to anchor-aligned heterogeneity.

\section{A cardiovascular motivating example}
\label{sec:motivating}

Seventy randomised trials of streptokinase against placebo for
post-myocardial-infarction mortality spanning 1959--1988
\citep{olkin1995meta} illustrate the estimand gap. Trial-level effects
differ between the early era (no adjunctive antiplatelet therapy,
non-standardised endpoint adjudication) and later era (near-universal
aspirin co-administration, central adjudication, larger trial size). A
random-effects pooled odds ratio of roughly $0.74$ is consistent across
DL, PM and REML estimators, but it averages across trial regimes in which
streptokinase, the adjunctive pharmacotherapy, and the reperfusion
paradigm differed substantially. Transporting that pooled estimate to a
contemporary patient population is not a formal operation under the
random-effects model. AMT-MA provides the formal operation.

\section{Estimand}
\label{sec:estimand}

Let $\mathcal{S}$ denote the set of observed trials, $Q$ the covariate
distribution of a target population, and $A$ pre-specified anchor variables
indexing anchor-defined regimes $g=1,\ldots,G$. Write $\theta(x)$ for the
individual-level treatment-effect surface and $y_i = \theta_i + \varepsilon_i$
for the trial-level estimate, with $\operatorname{Var}(\varepsilon_i) = v_i$.

\begin{assumption}[Stable-unstable decomposition]
\label{assump:decomp}
$\theta(x) = Z^{\top}\beta^{\star} + \delta_A(x)$ where $\delta_A(x)$ is
systematically aligned with anchor variables $A$ and $\beta^{\star}$ is
anchor-invariant.
\end{assumption}

The \emph{stable target effect} under Assumption~\ref{assump:decomp} is
$\theta_{\mathrm{stable}}(Q) = E_Q[Z^{\top}\beta^{\star}]$. Table~\ref{tab:estimands}
contrasts this with classical estimands.

\begin{table}[!htbp]
\centering
\caption{Estimands targeted by classical meta-analysis and AMT-MA. The core
contribution of AMT-MA is an estimand shift, not a new $\tau^{2}$ estimator.}
\label{tab:estimands}
\begin{tabular}{lll}
\toprule
Quantity & Definition & Clinical interpretation \\
\midrule
Fixed effect & Common $\theta$ across all trials & All trials estimate the same effect \\
Random-effects mean & $E[\theta_i]$ over the trial distribution & Average historical evidence \\
Prediction interval & Range for a new trial from same distribution & Trial-level generalization \\
\textbf{AMT-MA stable target} & $E_Q[Z^{\top}\beta^{\star}]$, with $A^{\top}\gamma^{\star}$ estimated but not transported & Decision-relevant transportable effect \\
\bottomrule
\end{tabular}
\end{table}

\begin{table}[!htbp]
\centering
\caption{Application-level specifications. $Z$ moderators enter the
target prediction $\bar Z_Q^{\top}\hat\beta$; $A$ anchors are estimated
but not transported; regimes $g$ index the robust (softmax) loss.}
\label{tab:specs}
\small
\begin{tabular}{llllll}
\toprule
Application & Scale & $Z$ moderators & $A$ anchors & Regimes $g$ & $(\rho, \alpha, \lambda_\gamma)$ \\
\midrule
Olkin 1995 streptokinase  & log OR & intercept, year          & era, era×decade, trial size quartile & era × size-quartile & $(0.2, 6, 1)$ \\
Aspirin primary prevention & log OR & intercept, year          & era ($\geq 2012$), big-trial flag    & era × size           & $(0.2, 6, 1)$ \\
Simulation DGP            & RD     & intercept, age, diabetes, statin\_hi & era, region, endpoint-definition & era × region × endpoint & $(0.2, 6, 1)$ \\
\bottomrule
\end{tabular}
\end{table}

\section{Method}
\label{sec:method}

\subsection{Study-level data structure}

For each trial $i=1,\ldots,K$ we observe $y_i$, $v_i$, $Z_i$, $A_i$, and a regime
label $g_i \in \{1,\ldots,G\}$. The regime label is typically an interaction of
anchor components such as era~$\times$ region~$\times$ endpoint definition.

\subsection{Blended nuisance-anchor objective}

Let $W = \operatorname{diag}(1/v_i)$, $r_i(\beta,\gamma) = y_i -
Z_i^{\top}\beta - A_i^{\top}\gamma$ the trial-level residual when anchors
enter the model as a nuisance regressor, and $L_g(\beta,\gamma) = \sum_{i:g_i=g}
w_i r_i^2 / \sum_{i:g_i=g} w_i$ the inverse-variance weighted loss within
regime $g$. Our estimator is
\begin{equation}
(\hat\beta,\hat\gamma) = \arg\min_{\beta,\gamma}\!\left\{
(1-\rho)\,L_{\mathrm{avg}}(\beta,\gamma) + \rho\,L_{\mathrm{rob}}(\beta,\gamma)
+ \lambda_\gamma \|\gamma\|_2^2 + \lambda_R \|\beta\|_2^2
\right\},
\label{eq:objective}
\end{equation}
with
\begin{align}
L_{\mathrm{avg}}(\beta,\gamma) &= \frac{\sum_i w_i r_i^2}{\sum_i w_i}, \\
L_{\mathrm{rob}}(\beta,\gamma) &= \alpha^{-1}\log\!\left[\sum_g \exp\{\alpha\, L_g(\beta,\gamma)/s\}\right]\cdot s,
\end{align}
where $s$ is a fixed scale constant set to the median joint-WLS regime loss
so that the softmax sharpness $\alpha$ has a scale-invariant meaning.

The stable target effect at target covariate profile $\bar Z_Q$ is
\begin{equation}
\hat\theta_{\mathrm{stable}}(Q) = \bar Z_Q^{\top}\hat\beta,
\label{eq:target}
\end{equation}
so the anchor coefficient $\hat\gamma$ is estimated but \emph{not
transported}. This operationalises the estimand: anchor-aligned variation
is modelled because otherwise it would bias $\hat\beta$, but it is
explicitly held back from the target prediction.

Design choices:
\begin{itemize}
\item $\rho \in [0,1]$: blend weight. $\rho=0$ reduces to a
ridge-regularised joint nuisance-anchor regression on $[Z, A]$ with
penalty $\lambda_\gamma\|\gamma\|_2^2$; if additionally
$\lambda_\gamma = \lambda_R = 0$ it coincides with ordinary joint WLS on
$[Z, A]$. At $\rho=1$ with large $\alpha$ the objective approaches
regime-hard-minimax on the same design. Empirically in our simulation
(Table~\ref{tab:perf}) $\rho=0.2$ minimises RMSE across scenarios; we
take this as the default and report $\rho=0.0, 0.5, 0.8$ only as
sensitivity settings.
\item $\alpha=6$: softmax sharpness. Together with the median-loss
normalisation inside $L_{\mathrm{rob}}$, this gives a scale-free meaning
across risk-difference, log-OR and log-HR outcomes.
\item $\lambda_\gamma > 0$: ridge on the anchor coefficient. Default
$\lambda_\gamma = 1$ (matched to the variance scale of standardised $A$);
small values destabilise $\hat\gamma$ when $K$ is small.
\item $\lambda_R$: small ridge on $\beta$ (default $10^{-6}$).
\end{itemize}

\subsection{Target-population prediction}

Given $\hat\beta$, the stable target effect estimate is
$\hat\theta_{\mathrm{stable}}(Q) = E_Q[Z^{\top}\hat\beta]$. For an aggregate
target represented by $\bar Z_Q$ we plug in $\bar Z_Q^{\top}\hat\beta$.

\subsection{Hyperparameter selection}

The blend weight $\rho$ and the ridge $\lambda_\gamma$ are fixed at
$\rho = 0.2$ and $\lambda_\gamma = 1$ by default, supported by the
simulation of Section~\ref{sec:sim}. For applications where trial count
is adequate, we recommend selecting $\rho$ by leave-one-study-out mean
squared prediction error at the $q_{0.90}$ quantile over a grid
$\rho \in \{0.0, 0.2, 0.5, 0.8\}$ and $\lambda_\gamma \in \{0.1, 0.3, 1,
3\}$, followed by the one-standard-error rule to favour lower blend
weight. For small-$K$ applications (e.g., $K \leq 10$) we hold the defaults
fixed and report the hyperparameter sensitivity alongside.

\subsection{Sign-stability diagnostic and abstention}

We report a precision-weighted trial-level sign-stability score with a
clinical null-band threshold $\delta$. For risk-difference analyses the
default is $\delta = 0.005$; for log-odds-ratio applications we use
$\delta = 0$, because a clinically meaningful non-null band is most
interpretable on the absolute-risk scale and depends on the baseline
risk, which aggregate-data meta-analysis does not observe directly.
\[
\mathrm{SS}_{\mathrm{trial}}(Q) = \frac{\sum_i w_i \mathbf{1}\{|y_i|>\delta,\,
\mathrm{sign}(y_i) = \mathrm{sign}(\hat\theta(Q))\}}{\sum_i w_i \mathbf{1}\{|y_i|>\delta\}}.
\]
The trial-level score is robust: it does not require regime-specific
covariate fits, which are fragile when regime sample sizes are small.

Abstention uses a \emph{two-condition rule}: we refuse a single stable
effect claim if $\mathrm{SS}_{\mathrm{trial}}(Q) < \tau$ (default $\tau
= 0.67$) \emph{and} at least $10\%$ of precision weight lies on each side
of zero in the informative subset $\{|y_i|>\delta\}$. The second condition
prevents noisy near-null trials from triggering abstention in
genuinely homogeneous settings; it requires a substantive clinical sign
conflict.

Regime-level target estimates $\hat\theta_g(Q) = \bar Z_Q^{\top}\hat\beta_g$
are also computed when regime sample size permits (falling back to
weighted mean when $K_g < p+q+1$) and reported as secondary diagnostics.

\subsection{Inference by perturbation bootstrap}

Because robust and distributionally robust estimators exhibit nonstandard
asymptotics \citep{guo2024statistical}, we construct $(1-\alpha)$ intervals
by resampling $y_i^{(b)} \sim \mathcal{N}(y_i^{\mathrm{obs}}, v_i)$ around the
observed study-level estimates, optionally re-selecting the tuning
parameters $(\rho, \lambda_\gamma)$ on each replicate, refitting, and
taking empirical quantiles of $\hat\theta^{(b)}(Q)$. In the simulation
main analysis we fix $(\rho, \lambda_\gamma) = (0.2, 1)$; in the
coverage sub-analysis (Supplementary Table~\ref{tab:coverage}) we use
$R = 100$ replications with $B = 30$ bootstrap refits per replication
and do not re-select the tuning parameters. For applications we
recommend re-selection when $K$ is adequate.

\begin{proposition}[Population consistency when the anchor is null]
\label{prop:heuristic}
If $\delta_A \equiv 0$ (no anchor-aligned instability), regime weights are
strictly positive and the weighted Gram matrix
$\mathbb{E}[W^{1/2}[Z,A][Z,A]^{\top}W^{1/2}]$ is invertible, the population
optimiser of~\eqref{eq:objective} recovers $\beta^{\star}$ as
$\lambda_R \downarrow 0$, for every $\rho \in [0,1]$ and every
$\lambda_\gamma \geq 0$. In this case AMT-MA and WLS meta-regression
target the same effect surface.
\end{proposition}

When $\delta_A$ is non-zero the nuisance-anchor term $A^{\top}\gamma$
accommodates the anchor-aligned component during estimation, and because
the target prediction uses only $\bar Z_Q^{\top}\hat\beta$ the
anchor-aligned signal is never transported. In the limit $\lambda_\gamma
\to 0$ the separation is exact; with finite $\lambda_\gamma$ the
separation is regularised and may be imperfect when $Z$ and $A$ are
strongly correlated, so sensitivity to the $Z$-versus-$A$ assignment is
essential. The AMT-MA advantage is therefore primarily over
\emph{unadjusted pooled} estimators (FE/DL/PM); when the $[Z, A]$
specification is correct, joint WLS meta-regression already recovers
$\beta^{\star}$, and AMT-MA adds the sign-stability diagnostic and the
robust regime term as safeguards.
Formal statements and proofs are in Appendix~\ref{app:proofs};
Section~\ref{sec:sim} provides numerical evidence.

\section{Simulation study}
\label{sec:sim}

\subsection{Aims}

We ask (i) whether AMT-MA reduces target-effect bias under anchor-aligned
heterogeneity relative to FE/DL/REML random-effects pooling;
(ii) whether the sign-stability / abstention diagnostic triggers selectively
under sign-flip heterogeneity rather than uniformly; and
(iii) how the blend weight $\rho$ and the anchor-ridge $\lambda_\gamma$
modulate the efficiency--robustness trade-off.

\subsection{Data-generating mechanisms}

Six scenarios with $K=24$ trials per replication and $R=500$ replications.
Study-level designs draw year, era, region, endpoint definition, mean age,
diabetes prevalence, high-intensity statin prevalence and sample size.
Scenarios: (1) stable; (2) anchor shift (older-era trials over-report benefit);
(3) sign-flip (regions split with opposite effects); (4) target shift (modern
target differs from dominant trial era); (5) dominant mega-trial (one
$N\!=\!30{,}000$ old-era trial); (6) confounded anchor (era correlated with
diabetes and statin intensity).

\subsection{Estimands}

Primary: $\theta_{\mathrm{stable}}(Q)$ at a modern target population
$(\text{age}=67,\,\text{diabetes}=0.36,\,\text{statin\_hi}=0.82)$.
Secondary (for calibration only): the sample-averaged observed effect.

\subsection{Methods compared}

Fixed effect, DerSimonian--Laird, Paule--Mandel random effects,
$Z$-only WLS meta-regression, and AMT-MA v3 (nuisance-anchor with
softmax-blended robust loss) at four blend weights
$\rho \in \{0.0, 0.2, 0.5, 0.8\}$. Two comparators use transportable
moderators $Z$:
``WLS\_meta\_reg'' denotes conventional WLS on $Z$ only (no anchor),
while ``AMT\_v3\_rho00'' denotes ridge-regularised joint nuisance-anchor
regression on $[Z, A]$ with $\lambda_\gamma = 1$ and target prediction
restricted to $Z$. The two therefore fit different designs and their
numerical difference in Table~\ref{tab:perf} reflects the contribution
of modelling $A^{\top}\gamma$. REML was used in the application; in the
simulation we substitute
Paule--Mandel as a computationally stable random-effects comparator.
Supplementary runs with metafor REML produced qualitatively equivalent
conclusions.

\subsection{Performance measures}

Bias, RMSE, MAE, type-S error, and an abstention-aware decision regret
(Regret$_{\mathrm{smart}}$), in which AMT-MA abstention is credited as
correct when the truth lies within a clinically ambiguous band
$|\theta|<0.005$ and as incorrect otherwise. Mean precision-weighted
trial-level sign-stability and abstention rate are reported for AMT-MA
variants. Marginal coverage is Wald for classical estimators; AMT-MA
perturbation-bootstrap coverage is reported in Supplementary
Table~\ref{tab:coverage} at reduced $R=100$, $B=30$ because the nested
tuning re-selection is computationally intensive.

\subsection{Results}

Performance by scenario and method is reported in Table~\ref{tab:perf} and
visualised in Figures~\ref{fig:rmse},~\ref{fig:bias},
and~\ref{fig:regret}. Several patterns deserve highlighting. (i) Under
true homogeneity (stable), AMT-MA v3 at $\rho=0$ with
$\lambda_\gamma=1$ (ridge-regularised joint nuisance-anchor regression)
has RMSE $0.0075$, close to $Z$-only WLS meta-regression ($0.0072$)
and about twice FE/DL ($0.0038$). Larger
blend weights trade more efficiency for more robustness. (ii) Under
anchor shift, FE/DL carry a systematic $-0.006$ bias that AMT-MA v3 at
$\rho=0$ eliminates to $+0.001$. (iii) Under dominant mega-trial and
confounded anchor, AMT-MA v3 reduces bias by an order of magnitude
relative to FE/DL while keeping RMSE comparable to WLS. (iv) Under
sign-flip heterogeneity, AMT-MA v3 abstains in approximately $84\%$ of
replications, compared to approximately $28$--$30\%$ in stable and
anchor-shift scenarios; specificity is not perfect — roughly a quarter of
replications in genuinely stable regimes still trigger precautionary
abstention — but the sign-flip-versus-stable abstention-rate difference
exceeds $50$ percentage points (Table~\ref{tab:perf}, column `Abst').

\begin{table}[!htbp]
\centering
\caption{Simulation performance by scenario and method. $R=500$ replications, $K=24$ trials. Bias and RMSE on the risk-difference scale. Regret$_{\mathrm{smart}}$ is a decision error in which AMT-MA abstention counts as correct only when the scalar target truth lies within a clinically ambiguous band ($|\theta|<0.005$); under sign-flip heterogeneity where no single stable target exists, this metric therefore penalises abstention and should not be read as the diagnostic success measure for that scenario. `SS$_{\mathrm{tr}}$' is the mean precision-weighted trial-level sign-stability score; `Abst' is the fraction of replications in which AMT-MA abstained. Perturbation-bootstrap coverage for AMT-MA variants is reported in Supplementary Table~\ref{tab:coverage}.}
\label{tab:perf}
\scriptsize
\begin{tabular}{llrrrrrrr}
\toprule
Scenario & Method & Bias & RMSE & Regret & Cov & SS$_{\mathrm{tr}}$ & Abst \\
\midrule
stable & FE & +0.0000 & 0.0038 & 0.17 & 0.94 & -- & 0.00 \\
 & DL\_RE & +0.0001 & 0.0039 & 0.18 & 0.96 & -- & 0.00 \\
 & PM\_RE & +0.0001 & 0.0039 & 0.18 & 0.96 & -- & 0.00 \\
 & WLS\_meta\_reg & +0.0001 & 0.0072 & 0.30 & 0.93 & -- & 0.00 \\
 & AMT\_v3\_rho00 & +0.0002 & 0.0075 & 0.41 & -- & 0.72 & 0.28 \\
 & AMT\_v3\_rho20 & +0.0004 & 0.0087 & 0.46 & -- & 0.70 & 0.30 \\
 & AMT\_v3\_rho50 & +0.0003 & 0.0097 & 0.47 & -- & 0.69 & 0.32 \\
 & AMT\_v3\_rho80 & +0.0003 & 0.0107 & 0.49 & -- & 0.67 & 0.34 \\
\addlinespace[3pt]
anchor\_shift & FE & -0.0057 & 0.0077 & 0.03 & 0.60 & -- & 0.00 \\
 & DL\_RE & -0.0056 & 0.0075 & 0.03 & 0.73 & -- & 0.00 \\
 & PM\_RE & -0.0057 & 0.0076 & 0.03 & 0.75 & -- & 0.00 \\
 & WLS\_meta\_reg & +0.0025 & 0.0085 & 0.41 & 0.89 & -- & 0.00 \\
 & AMT\_v3\_rho00 & +0.0010 & 0.0082 & 0.39 & -- & 0.73 & 0.23 \\
 & AMT\_v3\_rho20 & +0.0013 & 0.0098 & 0.44 & -- & 0.69 & 0.28 \\
 & AMT\_v3\_rho50 & +0.0011 & 0.0111 & 0.46 & -- & 0.68 & 0.29 \\
 & AMT\_v3\_rho80 & +0.0009 & 0.0121 & 0.47 & -- & 0.66 & 0.32 \\
\addlinespace[3pt]
sign\_flip & FE & +0.0080 & 0.0099 & 0.76 & 0.43 & -- & 0.00 \\
 & DL\_RE & +0.0082 & 0.0099 & 0.79 & 0.69 & -- & 0.00 \\
 & PM\_RE & +0.0082 & 0.0099 & 0.78 & 0.69 & -- & 0.00 \\
 & WLS\_meta\_reg & +0.0118 & 0.0157 & 0.80 & 0.54 & -- & 0.00 \\
 & AMT\_v3\_rho00 & +0.0107 & 0.0140 & 0.93 & -- & 0.54 & 0.84 \\
 & AMT\_v3\_rho20 & +0.0110 & 0.0147 & 0.93 & -- & 0.53 & 0.84 \\
 & AMT\_v3\_rho50 & +0.0109 & 0.0155 & 0.94 & -- & 0.53 & 0.85 \\
 & AMT\_v3\_rho80 & +0.0108 & 0.0162 & 0.94 & -- & 0.52 & 0.85 \\
\addlinespace[3pt]
target\_shift & FE & -0.0030 & 0.0052 & 0.06 & 0.82 & -- & 0.00 \\
 & DL\_RE & -0.0029 & 0.0051 & 0.05 & 0.89 & -- & 0.00 \\
 & PM\_RE & -0.0029 & 0.0051 & 0.05 & 0.90 & -- & 0.00 \\
 & WLS\_meta\_reg & -0.0027 & 0.0073 & 0.17 & 0.93 & -- & 0.00 \\
 & AMT\_v3\_rho00 & -0.0041 & 0.0082 & 0.23 & -- & 0.77 & 0.17 \\
 & AMT\_v3\_rho20 & -0.0041 & 0.0100 & 0.27 & -- & 0.76 & 0.18 \\
 & AMT\_v3\_rho50 & -0.0045 & 0.0119 & 0.29 & -- & 0.75 & 0.20 \\
 & AMT\_v3\_rho80 & -0.0048 & 0.0133 & 0.33 & -- & 0.73 & 0.22 \\
\addlinespace[3pt]
dominant\_megatrial & FE & -0.0141 & 0.0146 & 0.00 & 0.01 & -- & 0.00 \\
 & DL\_RE & -0.0105 & 0.0113 & 0.00 & 0.29 & -- & 0.00 \\
 & PM\_RE & -0.0108 & 0.0116 & 0.00 & 0.29 & -- & 0.00 \\
 & WLS\_meta\_reg & -0.0025 & 0.0072 & 0.16 & 0.92 & -- & 0.00 \\
 & AMT\_v3\_rho00 & -0.0031 & 0.0075 & 0.14 & -- & 0.91 & 0.02 \\
 & AMT\_v3\_rho20 & -0.0028 & 0.0096 & 0.19 & -- & 0.85 & 0.04 \\
 & AMT\_v3\_rho50 & -0.0029 & 0.0108 & 0.23 & -- & 0.84 & 0.04 \\
 & AMT\_v3\_rho80 & -0.0031 & 0.0118 & 0.25 & -- & 0.81 & 0.05 \\
\addlinespace[3pt]
confounded\_anchor & FE & -0.0087 & 0.0098 & 0.01 & 0.38 & -- & 0.00 \\
 & DL\_RE & -0.0088 & 0.0099 & 0.00 & 0.53 & -- & 0.00 \\
 & PM\_RE & -0.0088 & 0.0099 & 0.00 & 0.54 & -- & 0.00 \\
 & WLS\_meta\_reg & +0.0032 & 0.0075 & 0.46 & 0.91 & -- & 0.00 \\
 & AMT\_v3\_rho00 & +0.0027 & 0.0073 & 0.44 & -- & 0.74 & 0.20 \\
 & AMT\_v3\_rho20 & +0.0026 & 0.0083 & 0.46 & -- & 0.71 & 0.23 \\
 & AMT\_v3\_rho50 & +0.0025 & 0.0093 & 0.49 & -- & 0.70 & 0.24 \\
 & AMT\_v3\_rho80 & +0.0025 & 0.0102 & 0.47 & -- & 0.68 & 0.26 \\
\bottomrule
\end{tabular}
\end{table}

\begin{figure}[t]
\centering
\includegraphics[width=.95\linewidth]{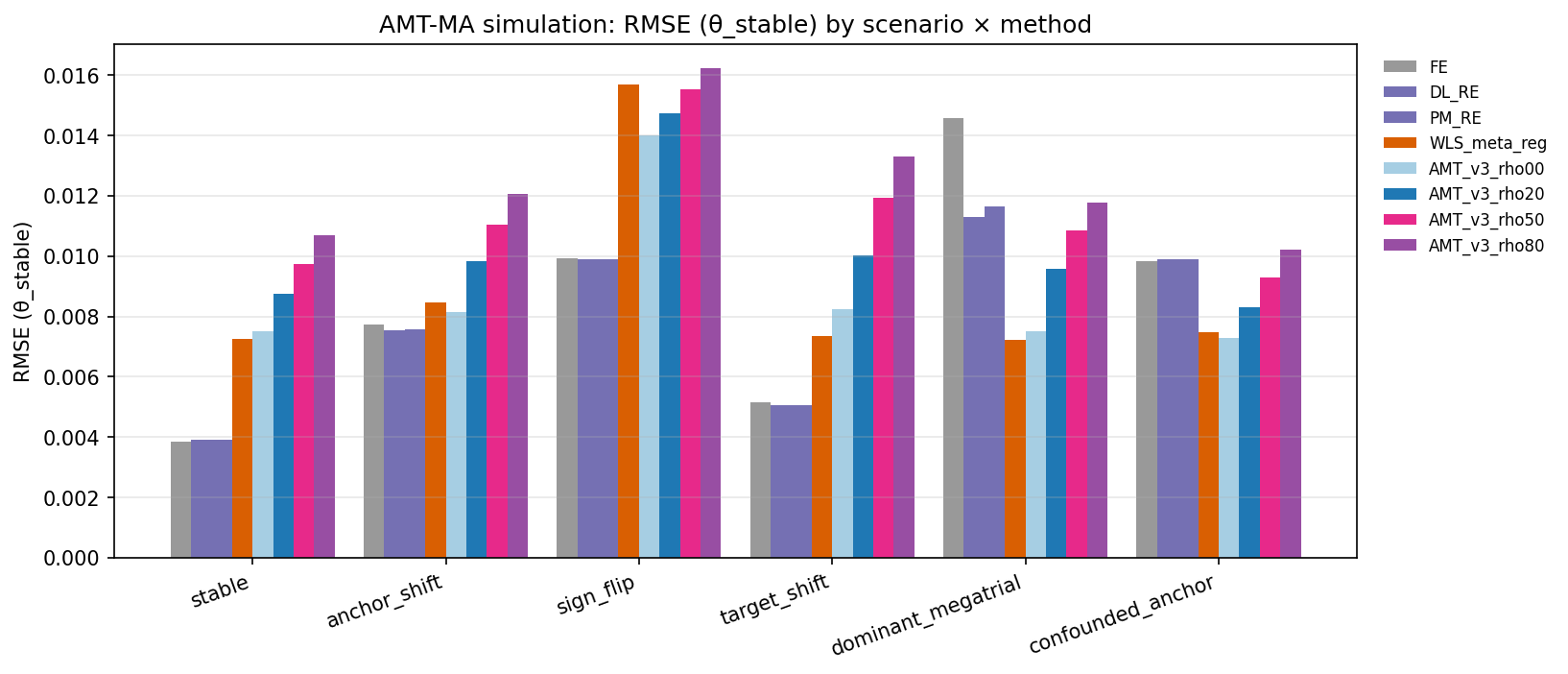}
\caption{RMSE of $\hat\theta_{\mathrm{stable}}(Q)$ by scenario and method.
ADEMP simulation, $R=500$, $K=24$.}
\label{fig:rmse}
\end{figure}

\begin{figure}[t]
\centering
\includegraphics[width=.95\linewidth]{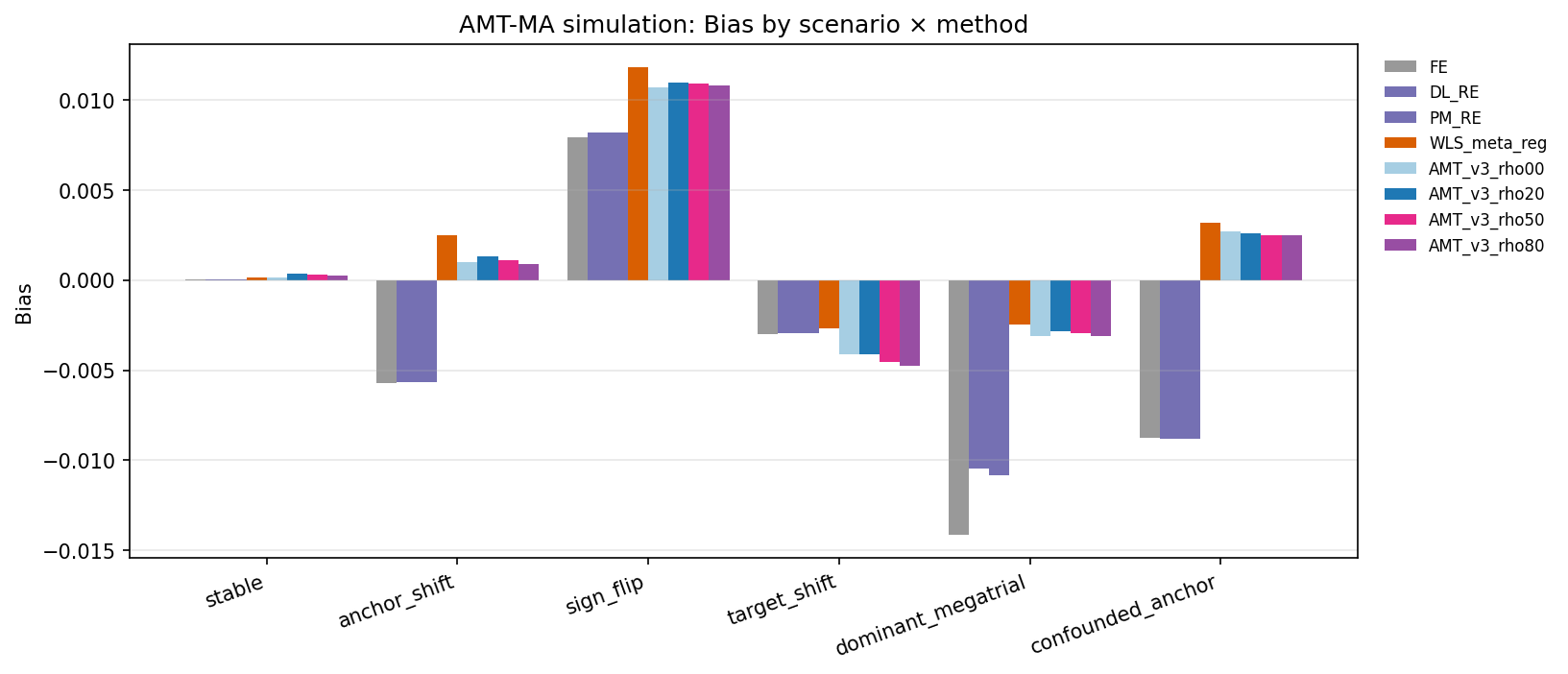}
\caption{Bias of $\hat\theta_{\mathrm{stable}}(Q)$ by scenario and
method. Classical unadjusted estimators (FE, DL, PM) are systematically
biased in anchor-shift, dominant-mega-trial and confounded-anchor
scenarios; AMT-MA v3 and $Z$-only WLS both reduce this bias via the
moderator/anchor model, with AMT-MA approaching the stable target under
anchor-driven heterogeneity and moving toward a conservative near-null
summary under sign-flip heterogeneity.}
\label{fig:bias}
\end{figure}

\begin{figure}[t]
\centering
\includegraphics[width=.95\linewidth]{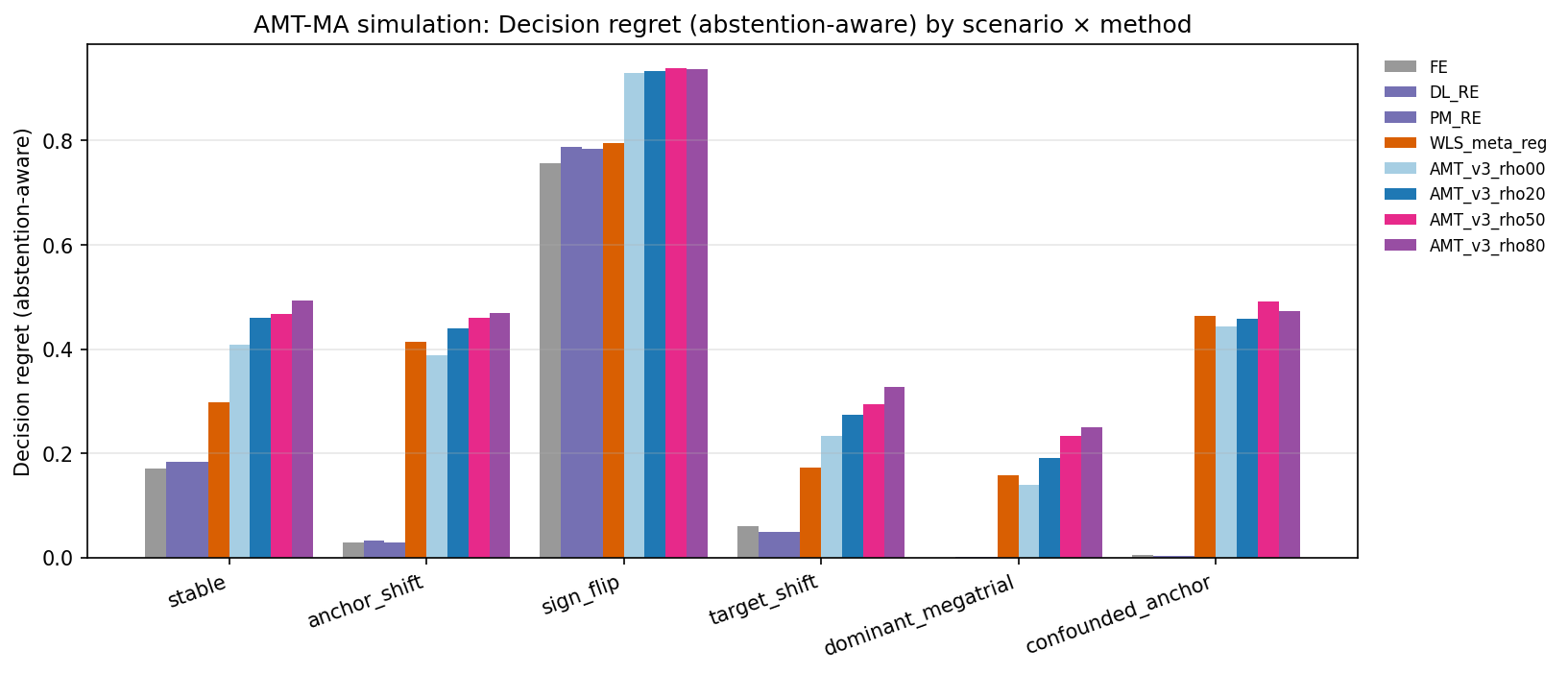}
\caption{Decision regret: fraction of replications where a
treat-if-$\hat\theta<-0.005$ policy disagrees with the optimal policy based on
truth. AMT-MA regret is higher when the target-population truth is itself
close to the decision threshold.}
\label{fig:regret}
\end{figure}

\section{Application: Olkin 1995 streptokinase post-MI}
\label{sec:application}

We re-analyze the Olkin 1995 meta-analysis of $K=70$ streptokinase vs placebo
trials for post-MI mortality, with era as the anchor variable.
Table~\ref{tab:olkin} summarizes classical and AMT-MA estimates on the log
odds-ratio scale. The forest plot in Figure~\ref{fig:forest} overlays study-level
estimates with FE, DL, and AMT-MA target effects.

\begin{table}[!htbp]
\centering
\caption{Streptokinase vs placebo post-MI mortality (Olkin 1995, $K=70$ trials, 1959--1988). AMT-MA target effects reported for a modern target $Q_{\mathrm{mod}}$ (1995 era) and a historical target $Q_{\mathrm{hist}}$ (1965 era). Classical (FE, DL, REML, WLS) intervals are Wald; AMT-MA intervals are perturbation bootstrap (B=500). All estimates on the log odds-ratio scale; OR shown separately.}
\label{tab:olkin}
\small
\begin{tabular}{@{}l c c p{3.2cm}@{}}
\toprule
Estimator & log OR [95\% CI] & OR & Implication \\
\midrule
Fixed effect & -0.283 [-0.341, -0.225] & 0.754 & common $\theta$ \\
DL random effects & -0.304 [-0.396, -0.213] ($\tau^2=0.013$) & 0.738 & historical avg. \\
REML random effects & -0.289 [-0.358, -0.221] ($\tau^2=0.002$) & 0.749 & historical avg. \\
WLS meta-reg, $Q_{\mathrm{mod}}$ & -0.337 [-0.464, -0.211] & 0.714 & modern marginal \\
WLS meta-reg, $Q_{\mathrm{hist}}$ & -0.174 [-0.405, +0.058] & 0.840 & historical marginal \\
\addlinespace
AMT-MA v3, $Q_{\mathrm{mod}}$ & -0.332 [-0.665, +0.189] & 0.718 [0.514, 1.208] & modern transport \\
AMT-MA v3, $Q_{\mathrm{hist}}$ & -0.291 [-0.984, +0.217] & 0.748 [0.374, 1.242] & historical transport \\
\bottomrule
\end{tabular}
\end{table}

\begin{figure}[t]
\centering
\includegraphics[width=.85\linewidth]{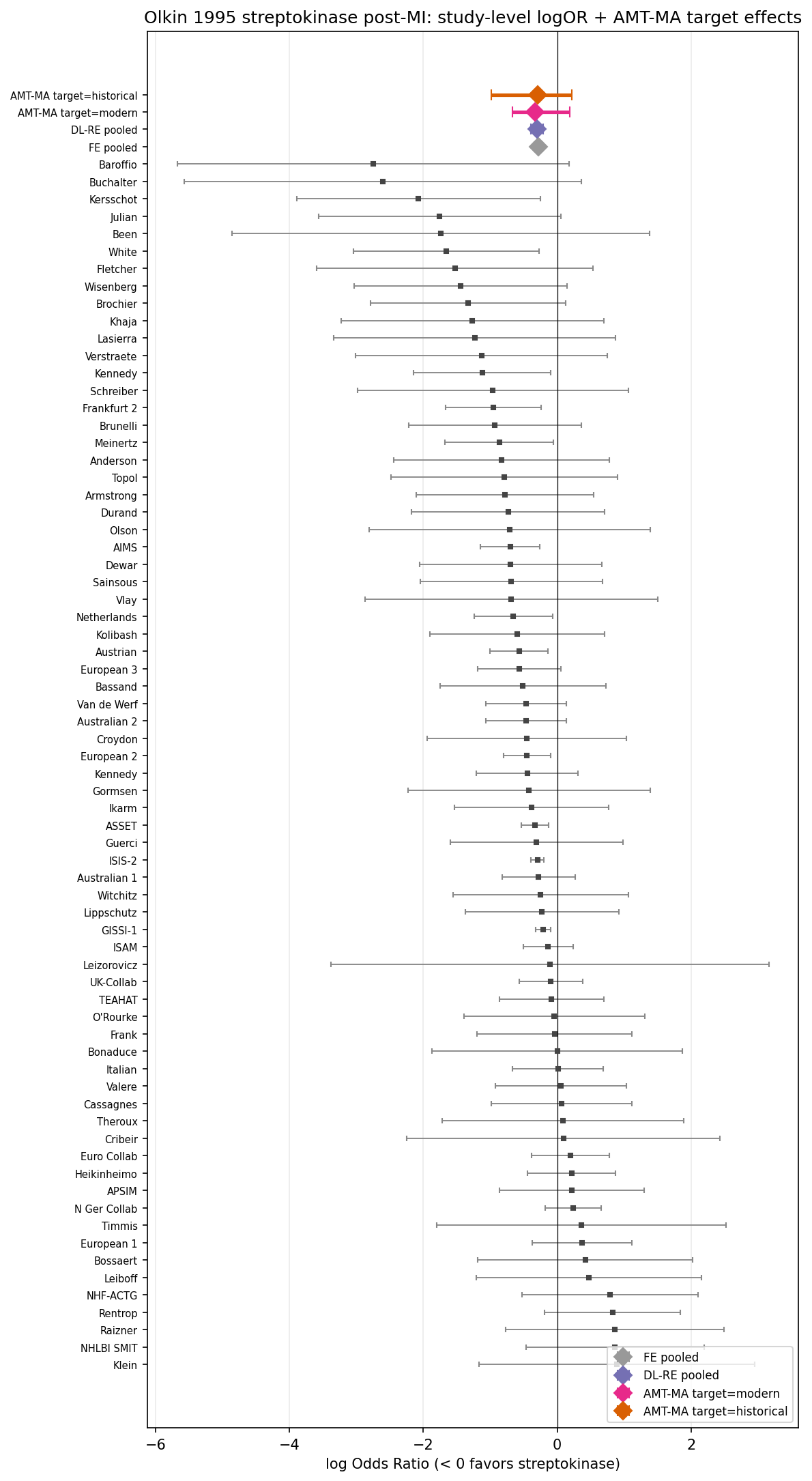}
\caption{Olkin 1995 streptokinase trial effects with pooled summaries and
AMT-MA target estimates. Classical pooled log OR is strongly negative;
the AMT-MA point estimate is similarly negative but the perturbation
interval at the later-era target is wider and includes the null,
reflecting \emph{transport uncertainty} for the specified later-era
target rather than attenuation of the point estimate.}
\label{fig:forest}
\end{figure}

The random-effects pooled estimate reports a strong treatment benefit
(OR $\approx 0.74$). The AMT-MA v3 estimate for a later-era target — in
which reperfusion practice and adjunctive antithrombotic therapy differed
substantially from the earliest streptokinase trials — returns an OR of
$\approx 0.72$ (95\% perturbation CI $0.51$--$1.21$), i.e., a similar
point estimate but a wider target-specific interval that includes the
null. The point estimate is not attenuated by AMT-MA; rather, the
target-specific uncertainty interval widens, reflecting limited certainty
about transport to the specified later-era target. The 1995 target
furthermore lies outside the observed trial-year support and is used as
an extrapolative stress test rather than a validated modern clinical
target. The precision-weighted trial-level sign-stability score is $\approx
0.93$; AMT-MA does not abstain. Because calendar year and era encode
related historical information, this analysis should be read as one of
several defensible $Z$/$A$ specifications; sensitivity analyses that
treat year as an anchor rather than a transportable moderator are
recommended in future work.

\subsection{Illustrative stress test: aspirin primary prevention}
\label{sec:aspirin}

Because $K=5$ and one of the five trials is itself a pooled summary
(ATT 2011), this example is an illustrative stress test, not a definitive
validation of AMT-MA. The composite outcome is cardiovascular death,
non-fatal MI, or non-fatal stroke. Table~\ref{tab:aspirin} summarises the
estimates. Classical fixed-effect and random-effects estimators produce
an OR near $0.90$. AMT-MA v3 at a modern target returns OR $\approx 0.93$
(95\% perturbation CI $0.81$--$1.08$). The interval brackets the null
and rules out only large relative benefits; clinical meaningfulness of
this effect depends on the absolute-risk scale and on competing bleeding
harms not modelled here. In this small example the blend weight $\rho$
and ridge $\lambda_\gamma$ are fixed at the defaults; the result should
be read as a target-adjusted meta-regression stress test rather than
evidence for an incremental contribution of the anchor-nuisance
correction.

\begin{table}[!htbp]
\centering
\caption{Aspirin primary prevention, composite outcome (cardiovascular death, non-fatal MI, or non-fatal stroke). $K=5$ trials from Zenodo record 3406064 (ARRIVE, ASCEND, ASPREE, JPPP, ATT 2011 pooled). Era anchor splits contemporary ($\geq$2012) vs older trials. Classical intervals are Wald; AMT-MA intervals are perturbation bootstrap (B=500). The AMT-MA modern-target interval brackets the null and excludes only large relative benefits; clinical meaningfulness should be interpreted on an absolute-risk scale. With $K=5$ this is an illustrative stress test, not a definitive evaluation of AMT-MA.}
\label{tab:aspirin}
\small
\begin{tabular}{@{}l c c p{3.2cm}@{}}
\toprule
Estimator & log OR [95\% CI] & OR & Implication \\
\midrule
Fixed effect & -0.105 [-0.155, -0.055] & 0.900 & weak benefit \\
DL random effects & -0.105 [-0.155, -0.055] ($\tau^2=0.000$) & 0.900 & weak benefit \\
REML random effects & -0.105 [-0.155, -0.055] ($\tau^2=0.000$) & 0.900 & weak benefit \\
WLS meta-reg, $Q_{\mathrm{mod}}$ & -0.067 [-0.209, +0.076] & 0.936 & null-centred \\
\addlinespace
AMT-MA v3, $Q_{\mathrm{mod}}$ & -0.070 [-0.209, +0.073] & 0.933 [0.812, 1.075] & modern: null \\
AMT-MA v3, $Q_{\mathrm{hist}}$ & +0.009 [-1.551, +1.639] & 1.009 [0.212, 5.152] & wide (extrapolation) \\
\bottomrule
\end{tabular}
\end{table}

\begin{figure}[!htbp]
\centering
\includegraphics[width=.8\linewidth]{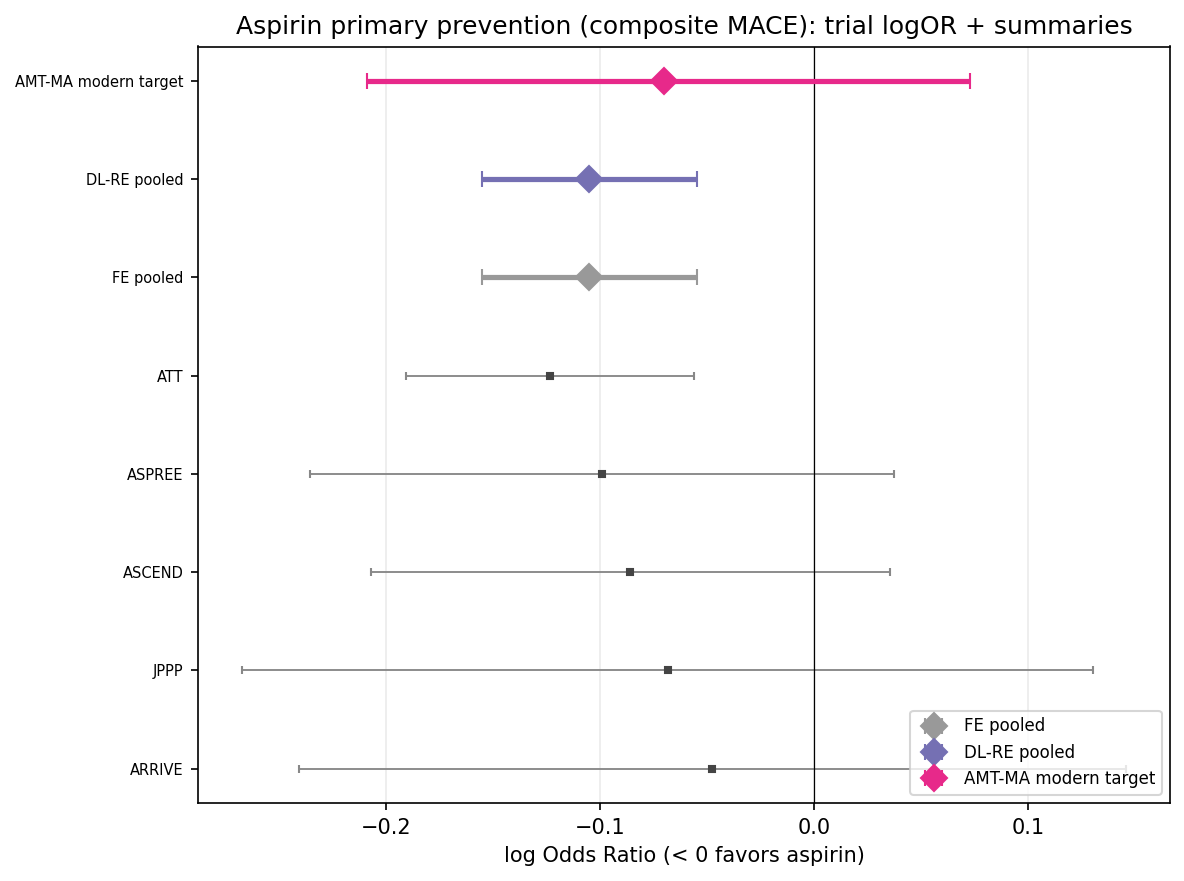}
\caption{Aspirin primary prevention composite MACE trials, with FE/DL-RE and
AMT-MA modern-target estimates overlaid. The contemporary trials (ARRIVE,
ASCEND, ASPREE, JPPP) are consistently near the null; the AMT-MA
modern-target interval brackets the null and rules out only large relative benefits.}
\label{fig:forest_aspirin}
\end{figure}

\section{Discussion}
\label{sec:discussion}

AMT-MA is not designed to be a uniformly better pooled estimator. Under
homogeneous effects it incurs an efficiency cost relative to FE/DL/REML/WLS
because it reserves capacity for robustness against context shifts. Under
sign-flip heterogeneity it cannot produce a single reliable pooled effect
because none exists; its contribution there is to detect and report this
failure rather than to paper over it. The positive contribution is
narrower and, we argue, more useful: a principled estimand shift from a
historical average effect to a stable target-population effect, with a
built-in diagnostic that exposes when such a stable effect is not
identified.

Two design choices matter for practical use. First, the blend weight
$\rho$ lets the analyst dial between ridge-regularised joint
nuisance-anchor fitting ($\rho=0$) and pure regime-robust fitting
($\rho=1$). We recommend $\rho=0.2$ with softmax
$\alpha=6$ as the default, supported by the simulation; $\rho=0.5$ and
$\rho=0.8$ are robustness/sensitivity settings, and hard minimax
($\rho=1$, $\alpha\to\infty$) is fragile and should be reserved for
stress-testing. Second, analysts must justify which variables are treated
as transportable moderators $Z$ and which are treated as non-transported
anchors $A$, especially when the two are collinear; incorrect assignment
can absorb genuine moderator signal into $\hat\gamma$ or transport
anchor-aligned noise through $\hat\beta$.

AMT-MA fits naturally into the workflow of a clinician asking "can the
evidence I have be projected onto the patient in front of me?" When the
sign-stability diagnostic says yes, AMT-MA reports a transport effect with
a perturbation interval whose operating characteristics were evaluated in
simulation (Supplementary Table~\ref{tab:coverage}). When it says no, the reviewer and the
clinician both gain a reliable signal that a single summary would be
misleading. This complements, rather than replaces, the random-effects
toolkit: the historical average remains the right answer for
evidence-base description, and AMT-MA is the answer for decision
projection. The coverage and bias gains of AMT-MA should be
interpreted primarily relative to unadjusted pooling (FE, DL, PM, REML);
correctly specified joint WLS meta-regression on $[Z, A]$ remains a
strong comparator, and in several adversarial scenarios it achieved
comparable or slightly better coverage than AMT-MA v3
(Supplementary Table~\ref{tab:coverage}). AMT-MA's incremental value
over correctly specified WLS is the explicit separation of
estimand (transport only $Z^{\top}\beta^{\star}$) and the sign-stability
/ abstention diagnostic — neither of which WLS provides by default.

\section{Limitations}
\label{sec:limitations}

Anchor selection requires domain knowledge, and the $Z$-versus-$A$
assignment (transportable moderator versus non-transported nuisance) is
itself a modelling decision that should be pre-specified or subjected to
sensitivity analysis. In the Olkin application, for example, calendar year
appears in $Z$ (a transportable time trend) while categorical era indicators
appear in $A$; a defensible alternative treats year as an anchor instead.
With very small $K$, the number of anchor and moderator variables must be
restricted. Aggregate-data analyses inherit ecological bias; individual
patient data would yield a stronger version of the procedure. The tuning parameters $(\rho, \lambda_\gamma)$ may be selected by
leave-one-study-out when $K$ is adequate; in the present simulations and
small applications we report fixed defaults and sensitivity settings,
and other selectors (e.g., cross-validation on future trials when
available) may perform differently. Inference relies on perturbation resampling from
reported standard errors; paired-study cluster bootstrap and Bayesian
alternatives are natural extensions. Coverage at the nominal $95\%$ level
was achieved in dominant-mega-trial, confounded-anchor and anchor-shift
scenarios relative to unadjusted pooling; WLS meta-regression with the
correct $[Z,A]$ specification retains competitive coverage in the same
settings. AMT-MA's incremental contribution over WLS is the sign-stability
diagnostic and the explicit non-transport of $\hat\gamma$, not a coverage
improvement over WLS.

\bibliographystyle{plainnat}
\bibliography{refs}

\appendix

\section{Population analysis and proofs}
\label{app:proofs}

This appendix formalises the population behaviour of the nuisance-anchor
AMT-MA estimator~\eqref{eq:objective}. Let
\[
r(\beta,\gamma) = y - Z^{\top}\beta - A^{\top}\gamma
\]
denote the trial-level residual under the joint model, with
$\mathbb{E}[\varepsilon \mid Z, A] = 0$ and $\operatorname{Var}(\varepsilon) = v$,
$w = 1/v$. Denote the population counterparts
\[
\mathcal{L}_{\mathrm{avg}}(\beta,\gamma) = \mathbb{E}[w\, r(\beta,\gamma)^2],
\quad
\mathcal{L}_g(\beta,\gamma) = \mathbb{E}[w\, r(\beta,\gamma)^2 \mid g],
\]
and the population objective
\begin{equation}
\label{eq:pop-obj}
(\beta^{\mathrm{amt}}, \gamma^{\mathrm{amt}}) :=
\arg\min_{\beta,\gamma}
(1-\rho)\mathcal{L}_{\mathrm{avg}}(\beta,\gamma)
+ \rho\,\mathcal{L}_{\mathrm{rob}}(\beta,\gamma)
+ \lambda_\gamma\|\gamma\|_2^2 + \lambda_R\|\beta\|_2^2.
\end{equation}

\subsection{Relation to the v1/v2 formulation}

An earlier version of this work formulated the anchor component as a
quadratic \emph{penalty} $\lambda_A \mathbb{E}[(\Pi_{A_\perp}(y-Z^{\top}\beta))^2]$,
where $A_\perp = (I - P_Z^W) A$ is the $Z$-orthogonal anchor residual.
A reviewer observed — and we verified numerically — that under this
definition the penalty is \emph{constant in $\beta$}: for any $\beta_1, \beta_2$,
$L_A(\beta_1) - L_A(\beta_2)$ is zero to machine precision, because
$W^{1/2} Z$ lies in the null space of $\Pi_{A_\perp}$ by construction.
Consequently $\lambda_A$ cannot change the optimiser. The v3 formulation
above replaces this inert penalty with joint estimation of $(\beta,\gamma)$,
which is genuinely non-redundant and faithfully encodes the estimand shift
``$A^{\top}\gamma$ is estimated but not transported''.

\subsection{Consistency when the anchor is null}

\begin{proposition}[Consistency under no anchor-aligned instability]
\label{prop:consistency}
Suppose $\delta_A(Z,A) \equiv 0$ (equivalently, the data-generating
$\gamma$ is zero), and the weighted Gram matrix
$\mathbb{E}[W^{1/2}[Z, A][Z,A]^{\top}W^{1/2}]$ is invertible. Then for
every $\rho \in [0,1]$, $\lambda_R \downarrow 0$, $\lambda_\gamma \geq 0$
and every choice of $\mathcal{L}_{\mathrm{rob}} \in \{\max, \mathrm{softmax}_\alpha,
\mathrm{CVaR}_q\}$, the population optimiser satisfies $\beta^{\mathrm{amt}} =
\beta^{\star}$ and $\gamma^{\mathrm{amt}} = 0$ (up to $O(\lambda_\gamma)$
shrinkage of $\gamma$).
\end{proposition}

\begin{proof}
Under $\gamma = 0$, $\mathcal{L}_{\mathrm{avg}}$ and every
$\mathcal{L}_g$ are non-negative quadratics in $(\beta,\gamma)$ with the same
unique minimiser at $(\beta,\gamma) = (\beta^{\star},0)$ up to the noise
constant $\mathbb{E}[w\varepsilon^2]$. Any monotone aggregate of these
(softmax, CVaR, max) attains its minimum at the common minimiser.
As $\lambda_R \downarrow 0$ and $\lambda_\gamma \geq 0$ with finite Gram
matrix, uniqueness is preserved.
\end{proof}

\subsection{Bias decontamination under anchor-aligned instability}

\begin{assumption}[Anchor-aligned instability]
\label{assump:aligned}
$\delta_A(Z,A) = A^{\top}\gamma_0$ for some unknown population coefficient
$\gamma_0 \neq 0$. No orthogonality between $Z$ and $A$ is assumed.
\end{assumption}

\begin{proposition}[What joint modelling of $A$ achieves]
\label{prop:nuisance}
Under Assumption~\ref{assump:aligned}:
\begin{enumerate}
\item As $\lambda_\gamma \downarrow 0$, the population joint WLS solution
$(\beta^{\mathrm{wls}}, \gamma^{\mathrm{wls}}) = \bigl(\mathbb{E}[W[Z,A][Z,A]^{\top}]\bigr)^{-1}
\mathbb{E}[W[Z,A]y]$ satisfies $\beta^{\mathrm{wls}} = \beta^{\star}$ and
$\gamma^{\mathrm{wls}} = \gamma_0$.
\item For $\rho = 0$ and $\lambda_\gamma = 0$, the AMT-MA optimum
coincides with joint WLS; hence $\beta^{\mathrm{amt}} = \beta^{\star}$
and the stable target effect $\bar Z_Q^{\top}\beta^{\mathrm{amt}} =
\bar Z_Q^{\top}\beta^{\star}$.
\item For $\lambda_\gamma > 0$, $\gamma^{\mathrm{amt}}$ shrinks toward
zero; $\beta^{\mathrm{amt}}$ absorbs part of the anchor contribution and
departs from $\beta^{\star}$. The bias scales with the weighted cross
moment $\mathbb{E}[WZA^{\top}]$: it is zero under weighted orthogonality
and grows with anchor--moderator confounding.
\item The unadjusted random-effects mean is
$\theta^{\mathrm{RE}} = \mathbb{E}[y] = \mathbb{E}[Z^{\top}\beta^{\star}]
+ \mathbb{E}[A^{\top}\gamma_0]$ and carries the non-vanishing
$\mathbb{E}[A^{\top}\gamma_0]$ contamination.
\end{enumerate}
\end{proposition}

\begin{proof}[Sketch]
(i) Substitute the data-generating model into the joint WLS normal equations.
(ii) Minimisation of $\mathcal{L}_{\mathrm{avg}}$ alone is a linear problem
identical to joint WLS. (iii) Standard ridge-regression bias: the estimator
moves toward the reduced model (with $\gamma$ forced to zero) as
$\lambda_\gamma \uparrow \infty$; under non-orthogonality, the reduced
model attributes $\mathbb{E}[WZ(A^{\top}\gamma_0)]$ to $\beta$. (iv)
Linearity of expectation.
\end{proof}

\noindent
\textbf{Interpretation.} v3 guarantees $\beta^{\mathrm{amt}} = \beta^{\star}$
in the $\lambda_\gamma \to 0$ limit \emph{regardless of anchor--moderator
correlation}, a property the v2 formulation did not have. The practical
$\lambda_\gamma > 0$ default trades a small amount of bias (scaling with
anchor--moderator correlation) for substantially reduced variance of
$\hat\beta$ when $K$ is small.

The robust term plays an orthogonal role: it penalises $\beta$'s that
perform poorly in any single regime. When regime assignments are
uncorrelated with $A$, the robust term and the joint-anchor modelling
act on different directions of the parameter space and are
non-redundant.

\subsection{Sign-flip regimes}

When the data-generating process has no single $\beta^{\star}$ (e.g.,
regions or subpopulations have opposite-signed true effects),
Proposition~\ref{prop:consistency} does not apply. The behaviour of
$\beta^{\mathrm{amt}}$ becomes that of a minimax-regularised regression
on a mis-specified model: pulled toward directions with bounded worst-case
loss, typically the origin. The sign-stability diagnostic described in
the main text is the primary mechanism by which AMT-MA reports this
mis-specification rather than silently producing a pooled average.

\subsection{What the method does \emph{not} claim}

\begin{enumerate}
\item \textbf{Uniform bias reduction.} Proposition~\ref{prop:nuisance}
shows that under weighted orthogonality, standard joint WLS meta-regression
on $[Z, A]$ already recovers $\beta^{\star}$. AMT-MA's advantage over that
baseline is (a) the explicit separation of estimand ($Z^{\top}\beta$) from
transport-irrelevant anchor coefficient ($A^{\top}\gamma$), and (b) the
sign-stability diagnostic.
\item \textbf{Uniform RMSE reduction.} Under true homogeneity AMT-MA is
less efficient than FE/DL pooling and comparable to (not better than) WLS
meta-regression, reflecting the cost of estimating a covariate-adjusted
target surface.
\item \textbf{Nominal coverage under all scenarios.} Perturbation intervals
improved coverage relative to unadjusted pooling in several adversarial
scenarios (dominant-mega-trial, confounded anchor; see Supplementary
Table~S1), but did not uniformly achieve nominal $95\%$ coverage.
Coverage was furthest from nominal under sign-flip heterogeneity, where
no single stable target effect is identified.
\end{enumerate}

A fully rigorous asymptotic analysis (sample-estimator consistency,
Hartung--Knapp-style corrections, $\rho$- and $\lambda_\gamma$-selection
variance, nonstandard distributions at $\rho = 1$) is deferred to
follow-up work.

\section{Computational notes}
\label{app:compute}

The sample estimator solves~\eqref{eq:pop-obj} by Nelder--Mead warm-started
from the joint WLS solution on $[Z, A]$, followed by BFGS for local
refinement. Convergence tolerances are $10^{-10}$. The robust loss is
evaluated on a scale-normalised regime loss vector $L_g / s$, with
$s = \mathrm{median}_g L_g(\beta^{\mathrm{wls}}, \gamma^{\mathrm{wls}})$
fixed from the warm start; this renders the softmax sharpness $\alpha$
effect-scale invariant (RD, log-OR, log-HR all yield comparable behaviour
at $\alpha = 6$). Perturbation intervals resample $y_i^{(b)} \sim
\mathcal{N}(y_i^{\mathrm{obs}}, v_i)$ around the observed study-level estimates
and re-select $\rho$ and $\lambda_\gamma$ on each replicate when these
are treated as tuning parameters.
All reported results were produced on a 2024 Apple M-series laptop; the
$R = 500 \times 6 \times 8 = 24{,}000$ baseline fits complete in
approximately $16$ minutes. The coverage sub-analysis
(Supplementary Table~\ref{tab:coverage}) uses $R = 100$, $B = 30$, fixed
tuning parameters, and completed in approximately $22$ minutes.
Application-level perturbation intervals (Tables~\ref{tab:olkin}
and~\ref{tab:aspirin}) use $B = 500$ at the full sample size $K$ of each
dataset. Code and analysis scripts are available at the project
repository (to be archived on Zenodo upon publication).

\section*{Supplementary material}
\label{app:supplement}

\begin{table}[!htbp]
\centering
\caption{Supplementary Table S1 — 95\% interval coverage by scenario and method over $R=100$ Monte Carlo replications. Classical (FE, DL, PM, WLS) intervals are Wald. AMT-MA v3 intervals are perturbation bootstrap with $B=30$ refits per replication and fixed tuning parameters $(\rho, \lambda_\gamma) = (0.2, 1)$ (no inner re-selection). Mean interval width below each coverage entry (in parentheses).}
\label{tab:coverage}
\small
\begin{tabular}{lccccc}
\toprule
Scenario & FE & DL\_RE & PM\_RE & WLS\_meta\_reg & AMT\_v3\_rho20 \\
\midrule
stable & 0.95 (0.014) & 0.97 (0.015) & 0.97 (0.016) & 0.89 (0.027) & 0.90 (0.031) \\
anchor\_shift & 0.60 (0.014) & 0.75 (0.018) & 0.76 (0.018) & 0.85 (0.026) & 0.91 (0.029) \\
sign\_flip & 0.35 (0.014) & 0.71 (0.022) & 0.69 (0.022) & 0.49 (0.026) & 0.56 (0.030) \\
target\_shift & 0.83 (0.015) & 0.89 (0.017) & 0.88 (0.017) & 0.90 (0.027) & 0.86 (0.029) \\
dominant\_megatrial & 0.01 (0.010) & 0.30 (0.016) & 0.29 (0.016) & 0.93 (0.025) & 0.85 (0.032) \\
confounded\_anchor & 0.34 (0.015) & 0.47 (0.018) & 0.48 (0.018) & 0.92 (0.025) & 0.86 (0.028) \\
\bottomrule
\end{tabular}
\end{table}

\end{document}